\documentclass[%
 reprint,
 amsmath,amssymb,
 aps,
]{revtex4-1}

\usepackage{graphicx}
\usepackage{dcolumn}
\usepackage{bm}
\usepackage{physics}
\usepackage{multirow}
\usepackage[hidelinks]{hyperref}
\usepackage{amsthm}
\usepackage{apptools}
\usepackage{chngcntr}
\usepackage{mathtools}
\usepackage{subcaption}
\usepackage{color}
\usepackage{tikz-cd}
\AtAppendix{\counterwithin{theorem}{section}}

\newtheorem{theorem}{Theorem}
\newtheorem{Theorem}[theorem]{Theorem}

\newtheorem{Lemma}[theorem]{Lemma}

\newtheorem{Corollary}[theorem]{Corollary}

\newtheorem{Definition}[theorem]{Definition}

\begin{document}

\title{Further Limitations on Information-Theoretically Secure Quantum Homomorphic Encryption}

\author{Michael Newman$^*$}%
\affiliation{%
Departments of Physics and Electrical and Computer Engineering\\
  Duke University, Durham, NC, 27708, USA
}%

\date{\today}

\begin{abstract}
In this brief note, we review and extend existing limitations on information-theoretically (IT) secure quantum fully homomorphic encryption (QFHE).  The essential ingredient remains Nayak's bound, which provides a tradeoff between the number of homomorphically implementable functions of an IT-secure QHE scheme and its efficiency.  Importantly, the bound is robust to imperfect IT-security guarantees. We summarize these bounds in the context of existing QHE schemes, and discuss subtleties of the imposed restrictions.
\end{abstract}
\maketitle

\section{Introduction}

Fully homomorphic encryption is one of the great advances of modern cryptography.  First discovered by Gentry in 2009 \cite{Gentry:2009}, it allows one to delegate the processing of \emph{encrypted} information to a party without access to the secret key.  In classical computing, an enormous body of work has gone into developing and optimizing this protocol (see \cite{Peikert:2015} for a summary).

As the development of large scale quantum computers progresses, we must consider the cryptographic consequences of their arrival.  While quantum computers can bolster the security of some cryptographic protocols \cite{Bennett:1984}, they can also obviate the security of others \cite{Shor:1995}.  Fortunately, the security of existing homomorphic encryption schemes is derived from hard problems on lattices \cite{Brakerski:2013, Gentry:2013}, which are expected to be computationally infeasible for quantum computers to solve.  Nonetheless, it is natural to ask: \vspace{.5 cm} \\ \emph{Can quantum computers allow for fully homomorphic encryption schemes which exhibit information-theoretic, rather than computational, security?} \vspace{.5cm}

For the strongest security definitions, we review why the answer is no while expanding on the limitations appearing in \cite{Newman:2017, Lai:2017, Yu:2014, Baumeler:2013}.

\section{Homomorphic Encryption}

\subsection{Classical homomorphic encryption}

A (classical) homomorphic encryption scheme is typically defined as an asymmetric key encryption scheme with an additional functionality, called evaluation.  This functionality allows a third party, in possession of a ciphertext, to meaningfully manipulate the underlying plaintext without possessing the secret key.  Formally, a homomorphic encryption scheme HE is a four-tuple of (randomized) algorithms.

\begin{itemize}
\item[] {\bf HE.KeyGen$(1^\kappa, 1^L)= (pk, sk, evk)$}. A key generation algorithm that accepts security parameter $\kappa$ and evaluation parameter $L$.  It outputs the public key $pk$ and secret key $sk$ which depend on $\kappa$.  Additionally, it outputs a third key known as the evaluation key $evk$, which depends on $L$.  This key assists with the additional evaluation functionality.

\item[] {\bf HE.Enc$(m,pk) = c$}. An encryption algorithm that accepts a public key $pk$ and a single bit plaintext $m$, and then outputs a ciphertext $c$.  By slight abuse of notation, we also allow $m$ to be a bit string, and assume the algorithm performs encryption bit-by-bit.

\item[] {\bf HE.Dec$(c,sk) = m$}. A decryption algorithm that accepts a single ciphertext $c$ and secret key $sk$ and outputs a single bit plaintext $m$.  Again, we allow multi-ciphertext inputs and assume decryption occurs bitwise.

\item[] {\bf HE.Eval$(C, (c_1, \ldots, c_n), evk) = c'$}. An evaluation algorithm that accepts a circuit $C$ with $n$ input wires.  It further accepts $n$ ciphertexts $(c_1, \ldots, c_n)$ and an evaluation key $evk$.  It outputs a single new ciphertext $c'$.
\end{itemize}

The encryption scheme HE should satisfy the usual properties of any encryption scheme, but should further satisfy the following homomorphic property.  For some circuits $C$ which we call the \emph{permissible functions} of the scheme, we have the following commutative diagram.

\begin{center}
\hspace{2cm}\begin{tikzcd}
\mathcal{M} \arrow[r, "\textbf{HE.Enc}{(\cdot,pk)}"]  \arrow{d}[swap]{C} &[1.5cm] \mathcal{C} \arrow[d, "\textbf{HE.Eval}{(C, \cdot, evk)}"] \\[1.5cm]
\mathcal{M} & \mathcal{C} \arrow[l, "\textbf{HE.Dec}{(\cdot, sk)}"]
\end{tikzcd}
\end{center}

Here, $\mathcal{M}$ is the space of valid plaintexts (i.e. all binary strings), and $\mathcal{C}$ is the space of valid ciphertexts.  One can think of the data processing as occurring from top-to-bottom, and the encryption as occurring from left-to-right.  Plainly, a party Alice can perform the computation of $C$ herself, or outsource the computation by sending an encrypted message to a third party Bob.

We further call a homomorphic encryption scheme \emph{compact} if the complexity of HE.Dec is independent of the function being evaluated.  This precludes trivial schemes in which Bob simply sends back a description of the circuit to be evaluated, and the true evaluation function is embedded into the decryption itself.  Here and throughout, we will only consider compact schemes.

We call a homomorphic encryption scheme \emph{leveled fully homomorphic} if the set of permissible functions of the scheme is the set of all circuits up to some size specified by the evaluation parameter $L$.  We simply call a scheme \emph{fully homomorphic} if the set of permissible functions is the set of all circuits, independent of $L$.  

Typically, the ciphertexts in homomorphic encryption schemes experience an accumulation of noise that scales with the evaluated circuit depth.  Eventually, this noise will prevent accurate decryption, and this motivates the definition of leveled fully homomorphic schemes.  It is important to note that Alice's work may scale with the circuit to be evaluated, but this is realized implicitly as preprocessing in the key generation phase.  A bootstrapping procedure introduced in \cite{Gentry:2009} allows for the indefinite refreshing of noisy ciphertexts, but requires a stronger circular security assumption.

Additionally, homomorphic encryption schemes sometimes allow for some small probability of failure.  To simplify the discussion, we demand that schemes are perfectly correct, but note that we can extend all our arguments to the imperfect case with some extra notational baggage \cite{Baumeler:2013}.

\subsection{IT-secure quantum homomorphic encryption}

In \cite{Broadbent:2014}, the problem of extending homomorphic encryption to the \emph{quantum} setting was considered.  Quantum homomorphic encryption accomplishes a similar task to classical homomorphic encryption, but with some key differences.  Formally, we can model QHE as three families of quantum channels parametrized by the input size $n$ acting on four Hilbert spaces: $\mathcal{H}_M$ the message space, $\mathcal{H}_K$ the key space, $\mathcal{H}_C$ the ciphertext space, and $\mathcal{H}_E$ the evaluated ciphertext space (see Figure \ref{HE_diagram}). 

\begin{figure}[htb!]
\includegraphics[width=\linewidth]{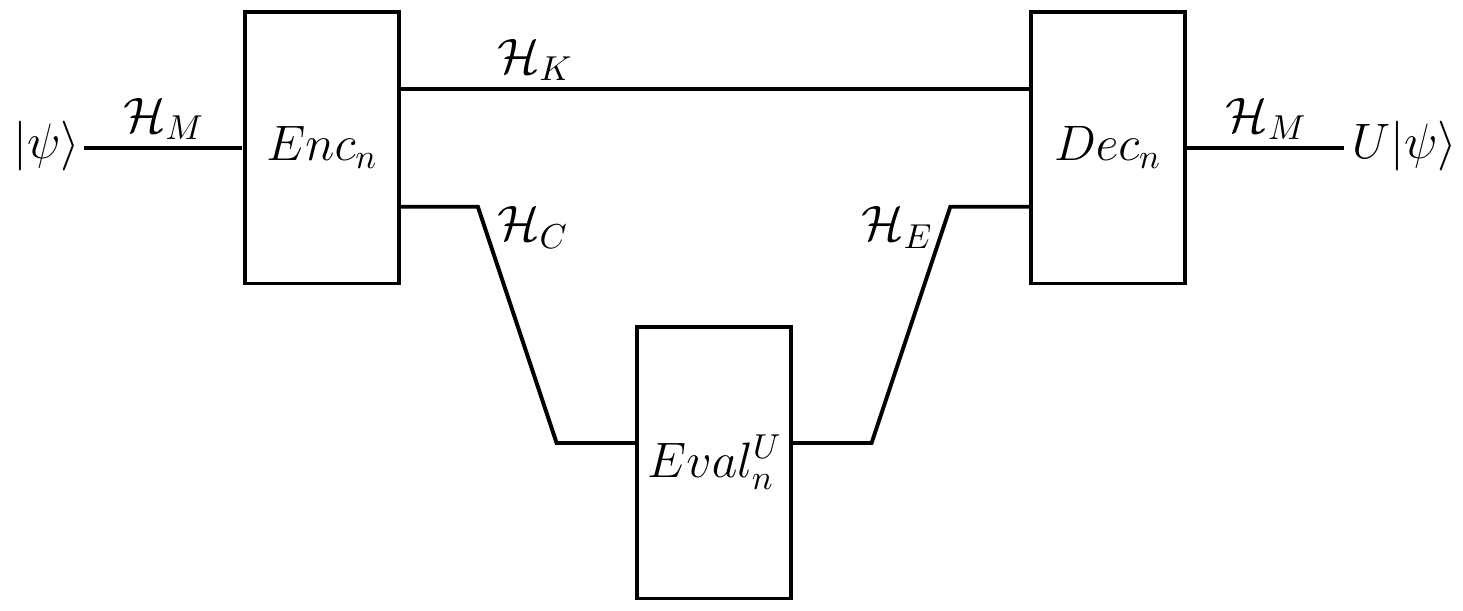}
\caption{A diagram of QHE matching the generality of the model in \cite{Yu:2014}.  Usually, we take $\mathcal{H}_C = \mathcal{H}_E$.}
\label{HE_diagram}
\end{figure}

Focusing on IT-secure encryption, we define a symmetric key QHE scheme.  Formally, it consists of the following families of channels.

\begin{itemize}

\item[] {\bf QHE.Enc$_n:L(\mathcal{H}_M) \longrightarrow L(\mathcal{H}_K \otimes \mathcal{H}_C)$}. An encryption isometry that accepts an input state $\ket{\psi}$ and then outputs a key $\rho_k \in L(\mathcal{H}_K)$ and a ciphertext $\rho_{c} \in L(\mathcal{H}_C)$ which includes an appended evaluation key \footnote{Note that a scheme defined with a randomly generated classical key can be transformed into a QHE scheme of the prescribed form by purifying the key.}.

\item[] {\bf QHE.Dec$_n:L(\mathcal{H}_K \otimes \mathcal{H}_C) \longrightarrow L(\mathcal{H}_M)$}. A decryption channel that accepts a key $\rho_k \in L(\mathcal{H}_K)$ and evaluated ciphertext $\rho_{e} \in L(\mathcal{H}_E)$, and returns  a plaintext state $\rho_m \in L(\mathcal{H}_M)$.

\item[] {\bf QHE.Eval$^U_n: L(\mathcal{H}_C) \longrightarrow L(\mathcal{H}_E)$}. An evaluation channel for a unitary circuit $U$ with $n$ wires that accepts a ciphertext state $\rho_c \in \mathcal{H}_C$ and outputs an evaluated ciphertext state $\rho_{e} \in \mathcal{H}_E$.  Usually $\mathcal{H}_C = \mathcal{H}_E$, but this is not required.

\end{itemize}

Notice that the security parameter $\kappa$, which determines the length of secret key and encoding size, does not appear explicitly.  Since quantum IT-secure encryption must obey similar key length bounds as classical IT-secure encryption \cite{lai2018generalized}, the key length must grow as a function of the input for fixed security guarantees.  Consequently, for a particular QHE scheme, we define $\kappa$ implicitly as some fixed polynomial of $n$.  Furthermore, encryption is performed all at once, rather than bit-by-bit.  

For leveled schemes, we similarly define the evaluation parameter $L$ as another fixed polynomial of $n$, and assume that the evaluation key is appended to the ciphertext.  This differs only slightly from the usual computationally-secure QHE setup \cite{Broadbent:2014}, but removes the need for explicit key generation.  

The homomorphic property of the scheme is defined analogously to the classical setting. There, homomorphic evaluation is usually defined piece-by-piece for the gates comprising a circuit.  For \emph{universal} gate sets, such as $\{$AND,OR$\}$ or $\{$NAND$\}$, homomorphic evaluation of these constituent gates can be built into (leveled) FHE schemes.  For QHE schemes, the set of permissible functions is augmented by a richer set of constituent gates. The definitions for compact and homomorphicity carry over analogously to the quantum setting.

\section{QHE proposals}

\subsection{Computationally secure proposals}
Computationally secure QHE was first considered in \cite{Broadbent:2014}, along with an appropriate generalization of CPA security.  The authors proposed three QHE schemes.  The first was a scheme that could homomorphically implement the set of Clifford circuits, and followed directly from the quantum one-time pad.  The second was a \emph{quasicompact} scheme: its decryption function scaled quadratically \emph{only} in the number of $T$-gates of the circuit.  The final scheme required an evaluation key size scaling superexponentially in the $T$-depth of the circuit, and so was restricted to efficiently evaluating circuits of constant $T$-depth.

More recently, \cite{Dulek:2016} proposed a scheme built off of previous work on instantaneous nonlocal computation \cite{Speelman:2015}. The scheme is centered around an evaluation key consisting of $T$-gadgets, which allows for the homomorphic evaluation of one $T$-gate per $T$-gadget.  This provides a leveled quantum fully homomorphic encryption scheme for polynomial-sized circuits, as these are the circuits for which the evaluation key can be generated efficiently.  This was further extended to a verifiable scheme in \cite{Dulek:2017}.  

Very recently, a scheme was proposed for performing leveled QFHE with a purely classical client using entirely different means \cite{Mahadev:2017}.  This scheme is built on an encrypted controlled-not function, which allows a server to apply the control-not gate blindly, i.e. without knowing whether they had or not.  The function requires a classical HE scheme with a particular list of desiderata. Whether or not this constitutes a non-leveled scheme depends on whether one accepts the circular security of existing schemes satisfying these requirements.

One common thread throughout all of these proposals is that each is built on a classical FHE scheme, and so inherits its underlying computational security.  It is natural to ask if quantum mechanics might allow for an \emph{information-theoretically} secure delegation of computation on encrypted information.  To this end, we might be encouraged by the construction of universal blind computation \cite{Broadbent:2008}, which allows delegated BQP quantum computation that guarantees information-theoretically secure hiding of \emph{both} the plaintext \emph{and} the computation, at the expense of interaction between Alice and Bob.

\subsection{Information-theoretically secure proposals}

There have been several works aimed towards homomorphic encryption with information-theoretic security guarantees.  In \cite{Tan:2014}, a homomorphic encryption scheme based on bosonic encodings was proposed.  This scheme used a weaker version of information-theoretic security by bounding the information accessible by the adversary.  This allowed them to realize a fully unitary group of permissible functions, although these were not universal as the dimension of the group scaled \emph{polynomially} in the input size.

Later, \cite{Ouyang:2015} proposed a homomorphic encryption scheme based off of randomized quantum codes and transversal gates.  The permissible functions for this scheme included all Clifford circuits augmented by a constant number of $T$-gates.  The security guarantees for this scheme were stronger, providing exponential suppression on the trace distance between any two ciphertexts.

More recently, \cite{Lai:2017} detailed a scheme based on the quantum one-time pad to implement the class of IQP circuits homomorphically.  Their scheme offers similarly strong information-theoretic security guarantees.

Finally, \cite{Tan:2017} proposed a homomorphic encryption scheme with a limited class of operations and modest information-theoretic security claims, but which may be implemented on current optical technologies.  For a broader survey of securely delegated quantum computing, see \cite{Fitzsimons:2017}.

\section{Restrictions on information-theoretically secure proposals}

We now elaborate on certain no-go theorems which limit the capacity of homomorphic encryption schemes to exhibit meaningful information-theoretic security.  It is well-known that in the classical setting, perfect information-theoretically secure homomorphic encryption is impossible \cite{Fillinger:2012}.  This follows from communication bounds established for perfectly secure single-server private information retrieval \cite{Kushilevitz:1997}, and their relaxations \cite{Shubina:2007}.

The first restriction on QHE information-theoretic security was proven in \cite{Yu:2014}.  There, they use a data localization argument via the no-programming theorem \cite{Nielsen:1997} to show the following.

\begin{Theorem}[Yu, Perez-Delgado, Fitzsimons] \label{first}
Suppose that a QHE scheme implements a set of unitary permissible functions $\mathcal{S}$, with precisely zero mutual information between the plaintext and ciphertext.  Then, the size of the evaluated ciphertext must be of size at least $\log_2(|\mathcal{S}|)$.
\end{Theorem}

When $S$ is the set of all classical reversible functions, this becomes $$\log_2((2^n)!) = (n - \log_2(e))2^n + O(n).$$ Thus, in the case of \emph{perfect} information-theoretic security, any QHE scheme implementing even just \emph{classical reversible} unitaries must be highly inefficient.  It is important to note that this data localization technique bounds the efficiency of perfect IT-secure QHE for \emph{any} set of unitary permissible functions.

In spite of this limitation, we have mentioned several QHE schemes \cite{Tan:2014,Ouyang:2015,Lai:2017} that implement large sets of permissible functions with strong, but imperfect, information-theoretic security guarantees.  As the information localization argument of \cite{Yu:2014} relies integrally on perfect information-theoretic security, several authors \cite{Yu:2014,Ouyang:2015,Aaronson:2017} have asked whether an $\epsilon$-relaxation of the IT-security guarantee may allow for much larger sets of permissible functions. Theorem \ref{earlier_main} resolves this question in the negative.  

\begin{Theorem} \label{earlier_main}
Suppose a QFHE scheme satisfies, for all ciphertexts $\rho,\rho'$, $$\|\rho - \rho'\|_{\text{Tr}} \leq \epsilon$$ for some $\epsilon < 1$.  Then the communication complexity of the scheme must be exponential in $n$.
\end{Theorem}

This theorem appeared concurrently in the appendix of \cite{Newman:2017}, where it was used to prove bounds on transversal gates for quantum codes, and in \cite{Lai:2017}, which used similar generalizations of single-server private information retrieval bounds to the quantum setting \cite{Baumeler:2013}.  At the heart of these arguments is an application of Nayak's bound \cite{Nayak:1999}, which places limitations on the compression of classical information into quantum information in the context of quantum random access codes.  We now extend Theorem~\ref{earlier_main} to more closely generalize Theorem~\ref{first} to the imperfect security setting by considering a straightforward modification of Nayak's original bound.

\begin{Definition} \label{QRAC}
An $(n,m,p)$\emph{-quantum random access code (QRAC)} is a mapping of $n$ classical bits into $m$ qubits, \emph{[$b \mapsto \rho_b$]}, along with a set of POVM's $\{M_i^0, M_i^1\}_{i=1}^n$ satisfying, for all $b\in \{0,1\}^n$ and $i\in[n]$, $$ Tr(M_i^{b_i}\rho_b) \geq p.$$
\end{Definition}

Informally, this is simply a compression of classical information into quantum information that allows for a local recovery with some probability of success.  Nayak's bound then places a fundamental limitation on the recoverability of this compression \cite{Nayak:1999}.  For Theorem \ref{main}, we will require a small generalization of the bound for QRACs encoding subsets of strings of length $n$.

\begin{Lemma} \label{main_lemma}
Let $F \subseteq \{0,1\}^n$.  For each $x \in F$, define a classical-to-quantum encoding $x \mapsto \rho_x$, where each $\rho_x$ is a state on $m$ qubits.  Suppose there exists a set of POVM's $\{M_i^0, M_i^1\}_{i=1}^n$ satisfying, for all $x \in F$ and $i\in[n]$, $$ Tr(M_i^{x_i}\rho_x) \geq p.$$
Let $H(\cdot)$ denote the binary entropy function.  Then, $$m \geq \log|F| - n\cdot H(p).$$
\end{Lemma}

\begin{proof}
The proof follows from the original proof of Nayak's bound, but with the uniform distribution on $F$ rather than $\{0,1\}^n$.  

For any $x \in \{0,1\}^*$, define $$F_x = \{y \in F| y \text{ is prefixed by } x\},$$ and let $\sigma_x = \frac{1}{|F_x|}\sum_{y \in F_x} \rho_y$.  For $0 \leq k < n$ and for any $x \in \{0,1\}^k$ such that $|F_x| > 0$, we then have $$\sigma_x = \frac{|F_{x0}|}{|F_x|}\sigma_{x0} + \frac{|F_{x1}|}{|F_{x}|} \sigma_{x1}.$$  Let $X^x$ be a random variable drawn from the uniform distribution on $F$ conditioned on the first $k$ bits being $x$, and let $Y^x$ be a binary random variable obtained from a measurement of $\sigma_x$.  Then by Holevo's theorem \cite{holevo1973bounds}, $$S(\rho_x) \geq \frac{|F_{x0}|}{|F_x|}S(\sigma_{x0}) + \frac{|F_{x1}|}{|F_{x}|} S(\sigma_{x1}) + I(X^x:Y^x).$$  If the probability of failure of correctly identifying a sample $\sigma_{x0}$ or $\sigma_{x1}$ conditioned on knowing the prefix $x$ is at most $1-p$, then we can apply Fano's inequality \cite{cover2012elements}, $$I(X^x:Y^x) \geq H\left(\frac{|F_{x0}|}{|F_{x}|}\right) - H(p).$$  We can then expand the bound for $S(\sigma)$ out into singletons $\rho_y$ so that, \begin{align*} S(\sigma) &\geq \frac{1}{|F|}\sum\limits_{y \in F}S(\rho_y) + \sum\limits_{k=0}^{n-1} \sum\limits_{x \in \{0,1\}^k} \frac{|F_x|}{|F|} \left(H\left(\frac{|F_{x0}|}{|F_x|}\right) - H(p) \right) \\ &\geq \sum\limits_{k=0}^{n-1} \sum\limits_{x \in \{0,1\}^k} \frac{|F_x|}{|F|}H\left(\frac{|F_{x0}|}{|F_x|}\right) - n\cdot H(p) \\ &\geq \sum\limits_{k=0}^{n-1} H(X_{k+1}|X_{1,\ldots,k}) - n\cdot H(p) \\ &\geq H(X) - n \cdot H(p) \end{align*} where the final line follows from the chain rule on conditional entropy.  Since $X$ is the uniform distribution on $|F|$ elements and $S(\sigma)$ is at most $m$, the result follows.

\end{proof}

With this lemma in mind, we can extend the definition of a QRAC to include encodings of subsets of strings, and refer to such encodings as $(F,n,m,p)$-QRACs, where $F$ is a subset of $\{0,1\}^n$.  Applying the lemma with $F = \{0,1\}^n$, we obtain Nayak's bound \cite{Nayak:1999}.

\begin{Theorem}[Nayak] \label{Nayak}
Any $(n,m,p)$-quantum random access code must satisfy $$m \geq n (1-H(p)).$$
\end{Theorem}

We now have the tools required to prove the extended imperfect IT-secure QFHE no-go theorem.  The central idea is to extract a quantum random access code from an IT-secure QHE scheme, and then apply Lemma \ref{main_lemma} to lower bound the communication complexity.

\begin{Theorem} \label{main}
Suppose we have a QHE scheme whose permissible functions can be used to evaluate a set $F_n$ of Boolean functions on $n$ bits.  Suppose further that, for any two ciphertexts $\rho, \rho'$ encrypting messages of $n$ bits, we have the $\epsilon$-IT security guarantee $$\| \rho - \rho' \|_{\text{Tr}} \leq \epsilon(n).$$ Then, the communication complexity of the protocol is at least $\log|F_n| - 2^n\cdot H(\epsilon(n)).$
\end{Theorem}

\begin{proof}
Throughout, we will use subscripts to denote the subsystem(s) on which an operator acts.  Fix any $n \geq 1$.  

For any $x \in \{0,1\}^n$, define the state $\rho_{KM}^x = \text{QHE.Enc}_n(\ket{x})$.  This is pure as QHE.Enc is an isometry.  Furthermore, because $\rho_M^x$ purifies to a state of dimension at most $2\dim(\mathcal{H}_M)$, we may assume that $\dim(\mathcal{H}_K) \leq \dim(\mathcal{H}_M)$.

By our IT-security guarantee, we have that for all $x,x'$, $$\|\rho_M^x - \rho_M^{x'}\|_{\text{\it Tr}} \leq \epsilon(n).$$  By the Schmidt decomposition, there must exist a unitary $U_K$ so that $$\|U_K \rho_{KM}^x U_K^\dag- \rho_{KM}^{x'}\|_{\text{\it Tr}} \leq \epsilon(n).$$  Fix any base point $x'$ and define $U_K^x$ to be an operator that satisfies the above equation.

Let $f:\{0,1\}^n \rightarrow \{0,1\}$ be any Boolean function in $F_n$.  To be explicit, we may assume $f$ is a classical reversible permissible function restricted to the first output wire, without loss of generality.  Consider $f$ as a length $2^n$ bit string with bit $i = f(i_{\text{binary}})$.  Then we can define the classical-to-quantum encoding $f \mapsto \rho_f$ according to $$\rho_f = \text{QHE.Eval}_n^{U_f}(\rho^{x'}_{KM})$$ where $U_f$ is the unitary representing the full reversible circuit for $f$.

Furthermore, for any $x \in \{0,1\}^n$, we may define the measurement channel $$M_x(\cdot) = M_{Z_1} \circ \text{QHE.Dec}_n \circ U_K^x(\cdot)U_K^{x\dag}$$ where $M_{Z_1}$ is a $Z$-basis measurement on the first qubit.  

We claim that this forms a $(S,2^n,\dim(\mathcal{H}_C) + \dim(\mathcal{H}_E), 1 - \epsilon(n))$-QRAC. To see this, note that \begin{align*} M_x(\rho_f) &= M_{Z_1} \circ \text{QHE.Dec}_n \circ U_K^x(\text{QHE.Eval}_n^{U_f}(\rho^{x'}_{KM}))U_K^{x\dag} \\ &= M_{Z_1} \circ \text{QHE.Dec}_n \circ (\text{QHE.Eval}_n^{U_f}(U_K^x\rho^{x'}_{KM}U_K^{x\dag}) \\ &\approx_{\epsilon(n)} M_{Z_1} \circ \text{QHE.Dec}_n \circ (\text{QHE.Eval}_n^{U_f}(\rho^{x}_{KM})) \\ &\approx_{\epsilon(n)} f(x),\end{align*}  where $\approx_{\epsilon(n)}$ denotes $\epsilon(n)$ proximity in terms of trace distance.  

In the above, the second line follows from the fact that evaluation is performed on the ciphertext alone, and so commutes with the keyspace unitary.  The third line follows from the IT-security guarantee and contractivity of trace distance.  The final line is exactly the homomorphic property of the scheme for $f \in F_n$.  

Thus, the probability of failure of the resulting QRAC is at most $\epsilon(n)$.  Recalling the assumption $\dim(\mathcal{H}_K) \leq \dim(\mathcal{H}_C)$, this gives the parameters of the claimed QRAC.  Applying Lemma \ref{main_lemma}, we see that $\dim(\mathcal{H}_C) + \dim(\mathcal{H}_E) \geq \log|S| - 2^n\cdot H(\epsilon(n))$.  Since $\dim(\mathcal{H}_C) + \dim(\mathcal{H}_E)$ is precisely the communication complexity of the protocol as a whole, the Theorem follows.
\end{proof}

It is worth noting that in the theorem, we only require encryptions of classical bits to be secure, since we are only concerned with classical functions.  Even in the case of hiding a \emph{constant} fraction of information, setting $F_n$ to be the set of all Boolean functions on $n$ bits, we obtain Theorem \ref{earlier_main} as a corollary.

\begin{proof}[Proof of Theorem~\ref{earlier_main}]
Because a QFHE scheme can implement the full set of Boolean functions, Theorem \ref{main} gives us that the communication complexity is lower bounded by $2^n(1 - H(\epsilon))$, which is exponential in $n$ whenever $\epsilon < 1$.
\end{proof}

More generally, one might expect that for most information-theoretically secure schemes, realizing security on the order of $2^{-poly(n)}$ could be augmented to security on the order of $O(2^{-1.01n})$ up to a polynomial change in the function defining the security parameter. For such a scheme, $2^nH(\epsilon(n)) = o(1)$, giving us the following corollary.

\begin{Corollary}
Suppose we have a QHE scheme whose permissible functions can be used to evaluate a set $F_n$ of $n$-bit Boolean functions.  Further suppose that for any pair of $n$-qubit ciphertexts $\rho,\rho'$, it satisfies $$\|\rho - \rho'\|_{\text{Tr}} \leq \epsilon(n)$$ for some $\epsilon(n) = O(2^{-1.01n})$.  Then, its communication complexity is lower bounded by $\log|S| - o(1)$.
\end{Corollary}

Note that an identical argument will still give a $\theta(\log|F_n|)$ lower bound even if $\epsilon(n) = \theta(2^{-n})$, assuming $|F_n|$ grows super-exponentially.  However, we reiterate that we would expect such schemes to be equivalent up to polynomially related security parameters.

In comparing Theorem \ref{main} to Theorem \ref{first}, we see that Theorem \ref{main} gives a nice generalization that matches the dependency on the number of classical permissible functions.  Most importantly, by reducing to lower bounds in quantum random access codes, \emph{it is robust to imperfect IT-security}.  The mild tradeoff for this reduction is that the lower bound is restricted to Boolean functions, rather than general unitaries.

\section{Discussion}

One thing that is of vital importance to notice: \emph{neither of these bounds rule out the possibility of IT-secure \emph{leveled} QFHE}.  Because the evaluation parameter is chosen as a fixed polynomial of the input, leveled schemes only carry the promise of implementing circuits whose size is some polynomial in the input.  A coarse counting argument shows that the number of such circuits can grow at most exponentially in the input size.  These information-theoretic approaches are ill-equipped to give a super-polynomial lower bound, simply because the class of permissible functions is too small.

To this end, we might turn to complexity theory as a means of proving no-go results.  Aaronson et al. have given complexity-theoretic evidence that perfectly-secure QFHE is impossible if it takes the natural but restricted form of an offline one-round generalized quantum encryption scheme \cite{Aaronson:2017}.  Even so, we re-emphasize that this relaxation to imperfect security is non-trivial: for example \cite{Aaronson:2017} shows that the delegation of blind quantum computation by classical clients is highly unlikely, but there already exist such schemes that leak a portion of the data \cite{mantri2017flow}.

Although we have required perfect correctness throughout for simplicity, the proofs go through using generalizations of Nayak's bound in the QPIR setting that account for imperfect retrievals \cite{Baumeler:2013}.  Furthermore, it would be nice to generalize the lower bound to account for unitary functions to fully generalize the result of \cite{Yu:2014}.

As with any no-go result, the most interesting question is how one might plausibly circumvent it.  In this direction, there are at least two avenues to consider.  

The first is to lessen the stringency of the security even further.  We've seen that, when using a weaker security guarantee in terms of accessible information, we can implement a full \emph{continuum} of permissible functions \cite{Tan:2014}.  Could something similar be made universal, and if not, how strong can this computational model be?

The second avenue to consider is to limit the number of implementable functions.  One could certainly restrict oneself to exponentially large circuit classes, as several existing IT-secure QHE schemes do.  Perhaps more interestingly, could one construct an efficient leveled QHE scheme with meaningful IT-security guarantees, and if so, what would such a scheme look like?  In existing IT-secure QHE schemes, for a fixed key size, information about the underlying plaintext is leaked as you encrypt more bits, but not as you apply more gates.  It seems that a leveled scheme would need to apply gates in such a way that the gates themselves leak information about the underlying plaintext; otherwise, the scheme might as well be non-leveled.  This idea is not without precedence, since evaluation key gadgets that both carry information about the secret key and apply gates homomorphically are a cornerstone of computationally-secure QFHE \cite{Dulek:2016}.  

Still, we reiterate that there is evidence that this cannot be done in the perfect-security setting.  In the direction of proving impossibility, perhaps further complexity-theoretic bounds on the set of permissible functions can be applied when restricting to a single-round of interaction \cite{Aaronson:2017}.

\section{Acknowledgements}
The author would like to thank Florian Speelman and Eric Sabo for comments on an earlier draft, as well as Cupjin Huang for pointing out Lemma~\ref{main_lemma}. This research was supported in part by NSF grant 1717523. \\ {\it $^*$Email address:} michael.newman@duke.edu.
\bibliography{bibliography.bib}

\begin{thebibliography}{31}%
\makeatletter
\providecommand \@ifxundefined [1]{%
 \@ifx{#1\undefined}
}%
\providecommand \@ifnum [1]{%
 \ifnum #1\expandafter \@firstoftwo
 \else \expandafter \@secondoftwo
 \fi
}%
\providecommand \@ifx [1]{%
 \ifx #1\expandafter \@firstoftwo
 \else \expandafter \@secondoftwo
 \fi
}%
\providecommand \natexlab [1]{#1}%
\providecommand \enquote  [1]{``#1''}%
\providecommand \bibnamefont  [1]{#1}%
\providecommand \bibfnamefont [1]{#1}%
\providecommand \citenamefont [1]{#1}%
\providecommand \href@noop [0]{\@secondoftwo}%
\providecommand \href [0]{\begingroup \@sanitize@url \@href}%
\providecommand \@href[1]{\@@startlink{#1}\@@href}%
\providecommand \@@href[1]{\endgroup#1\@@endlink}%
\providecommand \@sanitize@url [0]{\catcode `\\12\catcode `\$12\catcode
  `\&12\catcode `\#12\catcode `\^12\catcode `\_12\catcode `\%12\relax}%
\providecommand \@@startlink[1]{}%
\providecommand \@@endlink[0]{}%
\providecommand \url  [0]{\begingroup\@sanitize@url \@url }%
\providecommand \@url [1]{\endgroup\@href {#1}{\urlprefix }}%
\providecommand \urlprefix  [0]{URL }%
\providecommand \Eprint [0]{\href }%
\providecommand \doibase [0]{http://dx.doi.org/}%
\providecommand \selectlanguage [0]{\@gobble}%
\providecommand \bibinfo  [0]{\@secondoftwo}%
\providecommand \bibfield  [0]{\@secondoftwo}%
\providecommand \translation [1]{[#1]}%
\providecommand \BibitemOpen [0]{}%
\providecommand \bibitemStop [0]{}%
\providecommand \bibitemNoStop [0]{.\EOS\space}%
\providecommand \EOS [0]{\spacefactor3000\relax}%
\providecommand \BibitemShut  [1]{\csname bibitem#1\endcsname}%
\let\auto@bib@innerbib\@empty
\bibitem [{\citenamefont {Gentry}(2009)}]{Gentry:2009}%
  \BibitemOpen
  \bibfield  {author} {\bibinfo {author} {\bibfnamefont {C.}~\bibnamefont
  {Gentry}},\ }\href@noop {} {\enquote {\bibinfo {title} {A fully homomorphic
  encryption scheme},}\ } (\bibinfo {year} {2009}),\ \bibinfo {note} {{Ph.D.
  Thesis, Stanford University}}\BibitemShut {NoStop}%
\bibitem [{\citenamefont {Peikert}(2016)}]{Peikert:2015}%
  \BibitemOpen
  \bibfield  {author} {\bibinfo {author} {\bibfnamefont {C.}~\bibnamefont
  {Peikert}},\ }\href@noop {} {\enquote {\bibinfo {title} {A decade of lattice
  cryptography},}\ } (\bibinfo {year} {2016}),\ \bibinfo {note} {{Foundations
  and Trends in Theoretical Computer Science 10(4):283-424}}\BibitemShut
  {NoStop}%
\bibitem [{\citenamefont {{C. H. Bennett and G.
  Brassard}}(1984)}]{Bennett:1984}%
  \BibitemOpen
  \bibfield  {author} {\bibinfo {author} {\bibnamefont {{C. H. Bennett and G.
  Brassard}}},\ }\href@noop {} {\enquote {\bibinfo {title} {{Quantum
  cryptography: Public key distribution and coin tossing}},}\ } (\bibinfo
  {year} {1984}),\ \bibinfo {note} {{International Conference on Computers,
  Systems and Signal Processing, IEEE}}\BibitemShut {NoStop}%
\bibitem [{\citenamefont {Shor}(1997)}]{Shor:1995}%
  \BibitemOpen
  \bibfield  {author} {\bibinfo {author} {\bibfnamefont {P.~W.}\ \bibnamefont
  {Shor}},\ }\href@noop {} {\enquote {\bibinfo {title} {Polynomial-time
  algorithms for prime factorization and discrete logarithms on a quantum
  computer},}\ } (\bibinfo {year} {1997}),\ \bibinfo {note} {{SIAM
  J.Sci.Statist.Comput. 26, 1484}}\BibitemShut {NoStop}%
\bibitem [{\citenamefont {Brakerski}\ \emph {et~al.}(2013)\citenamefont
  {Brakerski}, \citenamefont {Langlois}, \citenamefont {Peikert}, \citenamefont
  {Regev},\ and\ \citenamefont {Stehlæ}}]{Brakerski:2013}%
  \BibitemOpen
  \bibfield  {author} {\bibinfo {author} {\bibfnamefont {Z.}~\bibnamefont
  {Brakerski}}, \bibinfo {author} {\bibfnamefont {A.}~\bibnamefont {Langlois}},
  \bibinfo {author} {\bibfnamefont {C.}~\bibnamefont {Peikert}}, \bibinfo
  {author} {\bibfnamefont {O.}~\bibnamefont {Regev}}, \ and\ \bibinfo {author}
  {\bibfnamefont {D.}~\bibnamefont {Stehlæ}},\ }\href@noop {} {\enquote
  {\bibinfo {title} {Classical hardness of learning with errors},}\ } (\bibinfo
  {year} {2013}),\ \bibinfo {note} {{Proceedings of STOC}}\BibitemShut
  {NoStop}%
\bibitem [{\citenamefont {Gentry}\ \emph {et~al.}(2013)\citenamefont {Gentry},
  \citenamefont {Sahai},\ and\ \citenamefont {Waters}}]{Gentry:2013}%
  \BibitemOpen
  \bibfield  {author} {\bibinfo {author} {\bibfnamefont {C.}~\bibnamefont
  {Gentry}}, \bibinfo {author} {\bibfnamefont {A.}~\bibnamefont {Sahai}}, \
  and\ \bibinfo {author} {\bibfnamefont {B.}~\bibnamefont {Waters}},\
  }\href@noop {} {\enquote {\bibinfo {title} {{Homomorphic Encryption from
  Learning with Errors: Conceptually-Simpler, Asymptotically-Faster,
  Attribute-Based}},}\ } (\bibinfo {year} {2013}),\ \bibinfo {note} {{Volume
  8042 of the series Lecture Notes in Computer Science pp. 75-92}}\BibitemShut
  {NoStop}%
\bibitem [{\citenamefont {Newman}\ and\ \citenamefont
  {Shi}(2018)}]{Newman:2017}%
  \BibitemOpen
  \bibfield  {author} {\bibinfo {author} {\bibfnamefont {M.}~\bibnamefont
  {Newman}}\ and\ \bibinfo {author} {\bibfnamefont {Y.}~\bibnamefont {Shi}},\
  }\href@noop {} {\enquote {\bibinfo {title} {Limitations on transversal
  computation through quantum homomorphic encryption},}\ } (\bibinfo {year}
  {2018}),\ \bibinfo {note} {{Q}uantum Information and Computation 18,
  927-948}\BibitemShut {NoStop}%
\bibitem [{\citenamefont {Lai}\ and\ \citenamefont {Chung}()}]{Lai:2017}%
  \BibitemOpen
  \bibfield  {author} {\bibinfo {author} {\bibfnamefont {C.-Y.}\ \bibnamefont
  {Lai}}\ and\ \bibinfo {author} {\bibfnamefont {K.-M.}\ \bibnamefont
  {Chung}},\ }\href@noop {} {\enquote {\bibinfo {title} {On
  statistically-secure quantum homomorphic encryption},}\ }\bibinfo {note}
  {{Q}uantum Information and Computation 18, 785-794}\BibitemShut {NoStop}%
\bibitem [{\citenamefont {Yu}\ \emph {et~al.}(2014)\citenamefont {Yu},
  \citenamefont {Perez-Delgado},\ and\ \citenamefont {Fitzsimons}}]{Yu:2014}%
  \BibitemOpen
  \bibfield  {author} {\bibinfo {author} {\bibfnamefont {L.}~\bibnamefont
  {Yu}}, \bibinfo {author} {\bibfnamefont {C.~A.}\ \bibnamefont
  {Perez-Delgado}}, \ and\ \bibinfo {author} {\bibfnamefont {J.~F.}\
  \bibnamefont {Fitzsimons}},\ }\href@noop {} {\enquote {\bibinfo {title}
  {Limitations on information theoretically secure quantum homomorphic
  encryption},}\ } (\bibinfo {year} {2014}),\ \bibinfo {note} {{Phys. Rev. A
  90, 050303}}\BibitemShut {NoStop}%
\bibitem [{\citenamefont {Baumeler}\ and\ \citenamefont
  {Broadbent}(2013)}]{Baumeler:2013}%
  \BibitemOpen
  \bibfield  {author} {\bibinfo {author} {\bibfnamefont {A.}~\bibnamefont
  {Baumeler}}\ and\ \bibinfo {author} {\bibfnamefont {A.}~\bibnamefont
  {Broadbent}},\ }\href@noop {} {\enquote {\bibinfo {title} {Quantum private
  information retrieval has linear communication complexity},}\ } (\bibinfo
  {year} {2013}),\ \bibinfo {note} {{Journal of Cryptology. Volume 28, Issue 1,
  pp 161-175 (2015)}}\BibitemShut {NoStop}%
\bibitem [{\citenamefont {Broadbent}\ and\ \citenamefont
  {Jeffery}(2015)}]{Broadbent:2014}%
  \BibitemOpen
  \bibfield  {author} {\bibinfo {author} {\bibfnamefont {A.}~\bibnamefont
  {Broadbent}}\ and\ \bibinfo {author} {\bibfnamefont {S.}~\bibnamefont
  {Jeffery}},\ }\href@noop {} {\enquote {\bibinfo {title} {Quantum homomorphic
  encryption for circuits of low \uppercase{T}-gate complexity},}\ } (\bibinfo
  {year} {2015}),\ \bibinfo {note} {{In Proceedings of Advances in Cryptology
  -- CRYPTO, pp 609-629}}\BibitemShut {NoStop}%
\bibitem [{Note1()}]{Note1}%
  \BibitemOpen
  \bibinfo {note} {Note that a scheme defined with a randomly generated
  classical key can be transformed into a QHE scheme of the prescribed form by
  purifying the key.}\BibitemShut {Stop}%
\bibitem [{\citenamefont {Lai}\ and\ \citenamefont
  {Chung}(2018)}]{lai2018generalized}%
  \BibitemOpen
  \bibfield  {author} {\bibinfo {author} {\bibfnamefont {C.-Y.}\ \bibnamefont
  {Lai}}\ and\ \bibinfo {author} {\bibfnamefont {K.-M.}\ \bibnamefont
  {Chung}},\ }\href@noop {} {\bibfield  {journal} {\bibinfo  {journal} {arXiv
  preprint arXiv:1801.03656}\ } (\bibinfo {year} {2018})}\BibitemShut {NoStop}%
\bibitem [{\citenamefont {Dulek}\ \emph {et~al.}(2016)\citenamefont {Dulek},
  \citenamefont {Schaffner},\ and\ \citenamefont {Speelman}}]{Dulek:2016}%
  \BibitemOpen
  \bibfield  {author} {\bibinfo {author} {\bibfnamefont {Y.}~\bibnamefont
  {Dulek}}, \bibinfo {author} {\bibfnamefont {C.}~\bibnamefont {Schaffner}}, \
  and\ \bibinfo {author} {\bibfnamefont {F.}~\bibnamefont {Speelman}},\
  }\href@noop {} {\enquote {\bibinfo {title} {Quantum homomorphic encryption
  for polynomial-sized circuits},}\ } (\bibinfo {year} {2016}),\ \bibinfo
  {note} {{CRYPTO: Advances in Cryptology, pp 3-32}}\BibitemShut {NoStop}%
\bibitem [{\citenamefont {Speelman}(2015)}]{Speelman:2015}%
  \BibitemOpen
  \bibfield  {author} {\bibinfo {author} {\bibfnamefont {F.}~\bibnamefont
  {Speelman}},\ }\href@noop {} {\enquote {\bibinfo {title} {Instantaneous
  non-local computation of low \uppercase{T}-depth quantum circuits},}\ }
  (\bibinfo {year} {2015}),\ \bibinfo {note} {{arXiv:1511.02839}}\BibitemShut
  {NoStop}%
\bibitem [{\citenamefont {{Gorjan Alagic and Yfke Dulek and Christian Schaffner
  and Florian Speelman}}(2017)}]{Dulek:2017}%
  \BibitemOpen
  \bibfield  {author} {\bibinfo {author} {\bibnamefont {{Gorjan Alagic and Yfke
  Dulek and Christian Schaffner and Florian Speelman}}},\ }\href@noop {}
  {\enquote {\bibinfo {title} {{Quantum Fully Homomorphic Encryption With
  Verification}},}\ } (\bibinfo {year} {2017}),\ \bibinfo {note} {{Proceedings
  of ASIACRYPT}}\BibitemShut {NoStop}%
\bibitem [{\citenamefont {Mahadev}(2017)}]{Mahadev:2017}%
  \BibitemOpen
  \bibfield  {author} {\bibinfo {author} {\bibfnamefont {U.}~\bibnamefont
  {Mahadev}},\ }\href@noop {} {\enquote {\bibinfo {title} {Classical
  homomorphic encryption for quantum circuits},}\ } (\bibinfo {year} {2017}),\
  \bibinfo {note} {{arXiv:1708.02130}}\BibitemShut {NoStop}%
\bibitem [{\citenamefont {Broadbent}\ \emph {et~al.}(2009)\citenamefont
  {Broadbent}, \citenamefont {Fitzsimons},\ and\ \citenamefont
  {Kashefi}}]{Broadbent:2008}%
  \BibitemOpen
  \bibfield  {author} {\bibinfo {author} {\bibfnamefont {A.}~\bibnamefont
  {Broadbent}}, \bibinfo {author} {\bibfnamefont {J.}~\bibnamefont
  {Fitzsimons}}, \ and\ \bibinfo {author} {\bibfnamefont {E.}~\bibnamefont
  {Kashefi}},\ }\href@noop {} {\enquote {\bibinfo {title} {Universal blind
  quantum computation},}\ } (\bibinfo {year} {2009}),\ \bibinfo {note}
  {{Proceedings of the 50th Annual IEEE Symposium on Foundations of Computer
  Science (FOCS), pp. 517-526}}\BibitemShut {NoStop}%
\bibitem [{\citenamefont {Tan}\ \emph {et~al.}(2016)\citenamefont {Tan},
  \citenamefont {Kettlewell}, \citenamefont {Ouyang}, \citenamefont {Chen},\
  and\ \citenamefont {Fitzsimons}}]{Tan:2014}%
  \BibitemOpen
  \bibfield  {author} {\bibinfo {author} {\bibfnamefont {S.-H.}\ \bibnamefont
  {Tan}}, \bibinfo {author} {\bibfnamefont {J.~A.}\ \bibnamefont {Kettlewell}},
  \bibinfo {author} {\bibfnamefont {Y.}~\bibnamefont {Ouyang}}, \bibinfo
  {author} {\bibfnamefont {L.}~\bibnamefont {Chen}}, \ and\ \bibinfo {author}
  {\bibfnamefont {J.~F.}\ \bibnamefont {Fitzsimons}},\ }\href@noop {} {\enquote
  {\bibinfo {title} {A quantum approach to homomorphic encryption},}\ }
  (\bibinfo {year} {2016}),\ \bibinfo {note} {{Sci. Rep. 6, 33467}}\BibitemShut
  {NoStop}%
\bibitem [{\citenamefont {Ouyang}\ \emph {et~al.}(2015)\citenamefont {Ouyang},
  \citenamefont {Tan},\ and\ \citenamefont {Fitzsimons}}]{Ouyang:2015}%
  \BibitemOpen
  \bibfield  {author} {\bibinfo {author} {\bibfnamefont {Y.}~\bibnamefont
  {Ouyang}}, \bibinfo {author} {\bibfnamefont {S.-H.}\ \bibnamefont {Tan}}, \
  and\ \bibinfo {author} {\bibfnamefont {J.}~\bibnamefont {Fitzsimons}},\
  }\href@noop {} {\enquote {\bibinfo {title} {Quantum homomorphic encryption
  from quantum codes},}\ } (\bibinfo {year} {2015}),\ \bibinfo {note}
  {{arXiv:1508.00938}}\BibitemShut {NoStop}%
\bibitem [{\citenamefont {{Si-Hui Tan and Yingkai Ouyang and Peter P.
  Rohde}}(2017)}]{Tan:2017}%
  \BibitemOpen
  \bibfield  {author} {\bibinfo {author} {\bibnamefont {{Si-Hui Tan and Yingkai
  Ouyang and Peter P. Rohde}}},\ }\href@noop {} {\enquote {\bibinfo {title}
  {{Practical quantum somewhat-homomorphic encryption with coherent states}},}\
  } (\bibinfo {year} {2017}),\ \bibinfo {note} {{arXiv:1710.03968}}\BibitemShut
  {NoStop}%
\bibitem [{\citenamefont {{Joseph F. Fitzsimons}}(2017)}]{Fitzsimons:2017}%
  \BibitemOpen
  \bibfield  {author} {\bibinfo {author} {\bibnamefont {{Joseph F.
  Fitzsimons}}},\ }\href@noop {} {\enquote {\bibinfo {title} {{Private quantum
  computation: an introduction to blind quantum computing and related
  protocols}},}\ } (\bibinfo {year} {2017}),\ \bibinfo {note} {{NPJ Quantum
  Information, Volume 3, Article number: 23}}\BibitemShut {NoStop}%
\bibitem [{\citenamefont {Fillinger}(2012)}]{Fillinger:2012}%
  \BibitemOpen
  \bibfield  {author} {\bibinfo {author} {\bibfnamefont {M.}~\bibnamefont
  {Fillinger}},\ }\href@noop {} {\enquote {\bibinfo {title} {Lattice based
  cryptography and fully homomorphic encryption},}\ } (\bibinfo {year}
  {2012}),\ \bibinfo {note} {{Master of Logic Thesis}}\BibitemShut {NoStop}%
\bibitem [{\citenamefont {{Eyal Kushilevitz and Rafail
  Ostrovsky}}(1997)}]{Kushilevitz:1997}%
  \BibitemOpen
  \bibfield  {author} {\bibinfo {author} {\bibnamefont {{Eyal Kushilevitz and
  Rafail Ostrovsky}}},\ }\href@noop {} {\enquote {\bibinfo {title}
  {{Replication is not needed: Single database, computationally-private
  information retrieval}},}\ } (\bibinfo {year} {1997}),\ \bibinfo {note}
  {{FOCS, pg. 364}}\BibitemShut {NoStop}%
\bibitem [{\citenamefont {{Amit Chakrabarti and Anna
  Shubina}}(2007)}]{Shubina:2007}%
  \BibitemOpen
  \bibfield  {author} {\bibinfo {author} {\bibnamefont {{Amit Chakrabarti and
  Anna Shubina}}},\ }\href@noop {} {\enquote {\bibinfo {title} {{Nearly Private
  Information Retrieval}},}\ } (\bibinfo {year} {2007}),\ \bibinfo {note}
  {{MFCS, pgs. 383-393}}\BibitemShut {NoStop}%
\bibitem [{\citenamefont {Nielsen}\ and\ \citenamefont
  {Chuang}(1997)}]{Nielsen:1997}%
  \BibitemOpen
  \bibfield  {author} {\bibinfo {author} {\bibfnamefont {M.~A.}\ \bibnamefont
  {Nielsen}}\ and\ \bibinfo {author} {\bibfnamefont {I.~L.}\ \bibnamefont
  {Chuang}},\ }\href@noop {} {\enquote {\bibinfo {title} {Programmable quantum
  gate arrays},}\ } (\bibinfo {year} {1997}),\ \bibinfo {note} {{Phys. Rev.
  Lett, 79, 321-4}}\BibitemShut {NoStop}%
\bibitem [{\citenamefont {Aaronson}\ \emph {et~al.}(2017)\citenamefont
  {Aaronson}, \citenamefont {Cojocaru}, \citenamefont {Gheorghiu},\ and\
  \citenamefont {Kashefi}}]{Aaronson:2017}%
  \BibitemOpen
  \bibfield  {author} {\bibinfo {author} {\bibfnamefont {S.}~\bibnamefont
  {Aaronson}}, \bibinfo {author} {\bibfnamefont {A.}~\bibnamefont {Cojocaru}},
  \bibinfo {author} {\bibfnamefont {A.}~\bibnamefont {Gheorghiu}}, \ and\
  \bibinfo {author} {\bibfnamefont {E.}~\bibnamefont {Kashefi}},\ }\href@noop
  {} {\enquote {\bibinfo {title} {On the implausibility of classical client
  blind quantum computing},}\ } (\bibinfo {year} {2017}),\ \bibinfo {note}
  {{arXiv:1704.08482}}\BibitemShut {NoStop}%
\bibitem [{\citenamefont {Nayak}(1999)}]{Nayak:1999}%
  \BibitemOpen
  \bibfield  {author} {\bibinfo {author} {\bibfnamefont {A.}~\bibnamefont
  {Nayak}},\ }\href@noop {} {\enquote {\bibinfo {title} {Optimal lower bounds
  for quantum automata and random access codes},}\ } (\bibinfo {year} {1999}),\
  \bibinfo {note} {{FOCS 1999}}\BibitemShut {NoStop}%
\bibitem [{\citenamefont {Holevo}(1973)}]{holevo1973bounds}%
  \BibitemOpen
  \bibfield  {author} {\bibinfo {author} {\bibfnamefont {A.~S.}\ \bibnamefont
  {Holevo}},\ }\href@noop {} {\bibfield  {journal} {\bibinfo  {journal}
  {Problemy Peredachi Informatsii}\ }\textbf {\bibinfo {volume} {9}},\ \bibinfo
  {pages} {3} (\bibinfo {year} {1973})}\BibitemShut {NoStop}%
\bibitem [{\citenamefont {Cover}\ and\ \citenamefont
  {Thomas}(2012)}]{cover2012elements}%
  \BibitemOpen
  \bibfield  {author} {\bibinfo {author} {\bibfnamefont {T.~M.}\ \bibnamefont
  {Cover}}\ and\ \bibinfo {author} {\bibfnamefont {J.~A.}\ \bibnamefont
  {Thomas}},\ }\href@noop {} {\emph {\bibinfo {title} {Elements of information
  theory}}}\ (\bibinfo  {publisher} {John Wiley \& Sons},\ \bibinfo {year}
  {2012})\BibitemShut {NoStop}%
\bibitem [{\citenamefont {Mantri}\ \emph {et~al.}(2017)\citenamefont {Mantri},
  \citenamefont {Demarie}, \citenamefont {Menicucci},\ and\ \citenamefont
  {Fitzsimons}}]{mantri2017flow}%
  \BibitemOpen
  \bibfield  {author} {\bibinfo {author} {\bibfnamefont {A.}~\bibnamefont
  {Mantri}}, \bibinfo {author} {\bibfnamefont {T.~F.}\ \bibnamefont {Demarie}},
  \bibinfo {author} {\bibfnamefont {N.~C.}\ \bibnamefont {Menicucci}}, \ and\
  \bibinfo {author} {\bibfnamefont {J.~F.}\ \bibnamefont {Fitzsimons}},\
  }\href@noop {} {\bibfield  {journal} {\bibinfo  {journal} {Physical Review
  X}\ }\textbf {\bibinfo {volume} {7}},\ \bibinfo {pages} {031004} (\bibinfo
  {year} {2017})}\BibitemShut {NoStop}%
\end{thebibliography}%

\end{document}